%% file: main.tex
\documentclass[submission,copyright,creativecommons]{eptcs}
\providecommand{\event}{HCVS 2016: 3rd Workshop on Horn Clauses for Verification and Synthesis} 
\usepackage{bm}
\usepackage{amsfonts}

\usepackage{amsmath}
\usepackage{amssymb}
\usepackage{amsthm}

\usepackage{centernot}
\usepackage{ifsym}
\usepackage{booktabs}
\usepackage{xstring}
\usepackage{marvosym}
\usepackage{diagbox}
\usepackage{caption}
\usepackage{subcaption}
\usepackage{graphicx}
\usepackage{makeidx}
\usepackage{float}

\usepackage{stmaryrd}
\usepackage{multicol}
\usepackage{multirow}
\usepackage{bigdelim}
\usepackage{pbox}
\usepackage{amsmath}
\usepackage{mathtools}
\usepackage{lipsum}
\usepackage[usenames,dvipsnames,svgnames,table]{xcolor}
\usepackage{xspace}
\usepackage{bytefield}
\usepackage{paralist}
\usepackage{minibox}
\usepackage[ruled,vlined]{algorithm2e}
\usepackage{leftidx}
\usepackage[mathcal]{euscript}
\usepackage{hyperref}
\usepackage{enumitem}
\usepackage{tabularx}
\usepackage{colortbl}
\usepackage{tikz}

\input{defs}

\hypersetup{
    bookmarks=true,         
    unicode=false,          
    pdftoolbar=true,        
    pdfmenubar=true,        
    pdffitwindow=false,     
    pdfstartview={FitH},    
    pdftitle={\mytitle},    
    pdfauthor={Paganelli},     
    pdfnewwindow=true,      
    colorlinks=true,       
    linkcolor=black,          
    citecolor=black,        
    filecolor=black,      
    urlcolor=black           
}

\definecolor{lstbg}{gray}{0.97}

\usepackage{listings}
\lstdefinestyle{grammar}{
  basicstyle=\fontsize{7pt}{1em}\it\sffamily,
  literate=
    {;}{{\textbf{;}}}{1}
    {¤oc}{{\bf{\{}}}{1}
    {¤cc}{{\bf{\}}}}{1}
    {¤=}{{\textbf{=}}}{1}
    {¤)}{{\textbf{)}}}{1}
    {¤(}{{\textbf{(}}}{1},
  captionpos=b,                   
  keywordstyle=\bfseries,
  morekeywords={=, A, Z, a, z, r0, 9, h, m, s, ms, us},
  escapeinside={@}{@},
  texcl = true,
  breaklines=false,                
}

\lstdefinestyle{pseudo}{
  captionpos=b,                   
  backgroundcolor=\color{white}
  keywordstyle=\bfseries,
  morekeywords={schedule,process,or},
  escapeinside={@}{@},
  texcl = true,
  frame=, framerule=1pt, rulecolor=\color{white},
  breaklines=false,                
}

\title{\mytitle}
\author{Gabriele Paganelli\\
\url{https://gapag.noblogs.org/}
\email{gapag@distruzione.org}
}

\def\titlerunning{\mytitle}
\def\authorrunning{G. Paganelli}
\begin{document}
\maketitle
\input{abstract}
\input{intro}
\input{background}
\input{deserializability}
\input{refinements}

\input{related}
\input{marskalk}
\bibliographystyle{eptcs}
\bibliography{biblio}
\end{document}

%% file: defs.tex
\newcommand{\Image}{\mathit{Image}}

\newcommand{\forward}{forward}
\newcommand{\backward}{backward}
\newcommand{\beginAx}{begin}

\newcommand{\join}{join}
\newcommand{\jumpLeft}{jumpLeft}
\newcommand{\jumpRight}{jumpRight}
\newcommand{\fieldAx}{field}
\newcommand{\ptrAx}{ptr}
\newcommand{\repAx}{repeat}
\newcommand{\repLenAx}{replen}

\newcommand{\reverse}{\mathit{reverse}}
\newcommand{\ascii}{\scName{ascii}}
\newcommand{\png}{\scName{png}}

\newcommand{\json}{\scName{json}}
\newcommand{\xml}{\scName{xml}}

\newcommand{\padsproj}{\scName{pads}}
\newcommand{\packettypes}{\scName{Packet Types}}
\newcommand{\datascript}{\scName{DataScript}}

\newcommand{\refForward}[1]{\refAxiom{\forward}{#1}}
\newcommand{\refBegin}{\texttt{\beginAx}}
\newcommand{\refBackward}[1]{\refAxiom{\backward}{#1}}
\newcommand{\refJoin}[1]{\refAxiom{\join}{#1}}
\newcommand{\refJumpLeft}[1]{\refAxiom{\jumpLeft}{#1}}
\newcommand{\refJumpRight}[1]{\refAxiom{\jumpRight}{#1}}

\newcommand{\refRepLenAx}[1]{\refAxiom{\repLenAx}{#1}}
\newcommand{\refAxiom}[2]{\texttt{#1}$_{#2}$}

\newcommand{\concatenate}{\centerdot}
\newcommand{\concatenateApp}[2]{#1\concatenate#2}

\newcommand{\forwardChainingInfer}{\mathit{forwardChainingInference}}

\newcommand{\tuple}[1]{\langle#1\rangle}
\newcommand{\alabeleditem}[1]{\tuple{#1}}

\newcommand{\simpleSuccessStruct}{\provingStructure{\axiomset',\infGraph}}
\newcommand{\simpleFailureStruct}{\provingStructure{\mathit{NonDeserializable}, \axiomset', \infGraph}}
\newcommand{\simpleFailureStructRet}{\provingStructure{\mathit{Deserializable}, \axiomset', \infGraph}}

\newcommand{\provingStructure}[1]{\tuple{#1}}

\newcommand{\tagForward}{\tagAxiom{\forward}{i}}
\newcommand{\tagBegin}{\tag{\refBegin}}
\newcommand{\tagBackward}{\tagAxiom{\backward}{i}}
\newcommand{\tagJoin}{\tagAxiom{\join}{i}}
\newcommand{\tagJumpLeft}{\tagAxiom{\jumpLeft}{o,s,i}}
\newcommand{\tagJumpRight}{\tagAxiom{\jumpRight}{o,s,i}}
\newcommand{\tagFieldAx}{\tagAxiom{\fieldAx}{i}}
\newcommand{\tagPtrAx}{\tagAxiom{\ptrAx}{i}}

\newcommand{\abarg}{\concatenateApp{b}{a}}
\newcommand{\tagRepForward}{\tagAxiom{\forward}{\abarg}}

\newcommand{\tagRepBackward}{\tagAxiom{\backward}{\abarg}}
\newcommand{\tagRepJoin}{\tagAxiom{\join}{\abarg}}
\newcommand{\tagRepJumpLeft}{\tagAxiom{\jumpLeft}{\abarg,s,i}}
\newcommand{\tagRepJumpRight}{\tagAxiom{\jumpRight}{\abarg,s,i}}

\newcommand{\tagRepAx}{\tagAxiom{\repAx}{b}}
\newcommand{\tagRepLenAx}{\tagAxiom{\repLenAx}{\abarg}}
\newcommand{\tagRepLenHead}{\tagAxiom{rephead}{\abarg}}
\newcommand{\tagRepLenTail}{\tagAxiom{reptail}{\abarg}}
\newcommand{\tagAxiom}[2]{\tag{\refAxiom{#1}{#2}}}
\newcommand{\padparsergraph}{\def\arraystretch{0.3}\setlength{\tabcolsep}{0.2em}}
\newcommand{\scName}[1]{\textsc{#1}\xspace}
\newcommand{\fieldId}[1]{\lstinline{#1}\xspace}
\newcommand{\pngLength}{\fieldId{Length}}
\newcommand{\pngType}{\fieldId{Type}}
\newcommand{\pngData}{\fieldId{Data}}
\newcommand{\pngCRC}{\fieldId{CRC}}
\newcommand{\occurs}{\hookrightarrow}
\newcommand{\noccurs}{\centernot{\hookrightarrow}}
\newcommand{\LabeledItems}{\Phi}
\newcommand{\LabeledRepItems}{P}
\newcommand{\headpos}{\nearrow}
\newcommand{\headpostab}{\multicolumn{1}{r}{$\headpos$}}
\newcommand{\headpostabbuf}{\multicolumn{1}{r}{\bufferized$\headpos$}}
\newcommand{\bufferized}{\cellcolor[gray]{0.9}}
\newcommand{\clips}{\scName{CLIPS}}
\newcommand{\python}{\scName{Python}}
\newcommand{\jess}{\scName{Jess}}

\newcommand{\items}{\mathcal{I}}
\newcommand{\repeitems}{\mathcal{R}}
\newcommand{\anitem}{\iota}
\newcommand{\anotheritem}{\upsilon}
\newcommand{\myanitem}{\anitem_{\interpr{i}}}
\newcommand{\myanitemarg}[1]{\anitem_{\interpr{#1}}}
\newcommand{\interpr}[1]{\llbracket#1\rrbracket}
\newcommand{\ptrRange}{\mathit{range}}
\newcommand{\reader}{parser\xspace}
\newcommand{\Reader}{Parser\xspace}
\newcommand{\landingpad}{o}
\newcommand{\premise}{\alpha}
\newcommand{\consequence}{\beta}
\newcommand{\pointerspan}{s}
\newcommand{\field}{\mathbf{f}}
\newcommand{\const}{\mathbf{c}}
\newcommand{\varfield}{\mathbf{v}}

\newcommand{\mytitle}{Horn Binary Serialization Analysis}
\newcommand{\readerFol}[1]{\mathfrak{R}_{#1}}
\newcommand{\readerRepFol}[1]{\mathfrak{S}_{#1}}
\newcommand{\folpredicate}[1]{\mathit{#1}}
\newcommand{\beg}[1]{\folpredicate{Beg}(#1)}
\newcommand{\len}[1]{\folpredicate{Len}(#1)}
\newcommand{\val}[1]{\folpredicate{Val}(#1)}
\newcommand{\rep}[1]{\folpredicate{Rep}(#1)}

\newcommand{\true}{\mathit{true}}
\newcommand{\KB}{\mathit{KB}}
\newcommand{\infGraph}{\mathit{G}}

\newcommand{\modusponens}{\vdash}
\newcommand{\axiomset}{\mathcal{A}}
\newcommand{\formalSystem}{\mathfrak{D}}

\newcommand{\pointerFieldNum}{\pi
}

\newcommand{\repetition}[1]{\bm{[}#1\bm{]*}}
\newcommand{\repLen}[1]{\folpredicate{RepLen}(#1)}
\newcommand{\pspan}[1]{\folpredicate{Ptr}(#1)}

\newcommand{\pointer}[2]{\llparenthesis#1\rightslice#2\rrparenthesis}
\newcommand{\Pointers}{\{\pointer{\landingpad}{\pointerspan} | \landingpad, \pointerspan \in \naturals\}}
\newcommand{\EnhancedPointers}{\{\pointer{\landingpad}{\pointerspan} | \landingpad \in \naturals^*, \pointerspan \in \naturals\}}
\newcommand{\pointers}{\mathcal{B}}
\newcommand{\Layouts}{\mathcal{L}}
\newcommand{\RepeatLayouts}{\mathcal{M}}
\newcommand{\streams}{\mathcal{S}}

\newcommand{\alayout}{\vec{\ell}}
\newcommand{\itemarr}{\vec{r}}
\newcommand{\synConstants}{\mathcal{C}}
\newcommand{\synPreds}{\mathcal{P}}
\newcommand{\synFuncs}{\mathcal{F}}
\newcommand{\synVars}{\mathcal{V}}
\newcommand{\labelfun}{\lambda}
\newcommand{\labelfunrep}{\bm{\mu}}
\newcommand{\emptyLayout}{\epsilon}

\def\ReplaceStr#1#2#3{%
  \IfSubStr{#1}{#2}{%
    \StrSubstitute{#1}{#2}{#3}}{#1}}

\newcommand{\intlist}[1]{
\ReplaceStr{#1}{,}{\concatenate}
}

\newcommand{\treetoflat}{\nabla}

\newcommand{\alabel}{\omega}

\newcommand{\erlang}{\scName{Erlang}}
\newcommand{\haskell}{\scName{haskell}}

\newcommand{\prolog}{\scName{Prolog}}

\newcommand{\naturals}{\mathbb{N}}
\newcommand{\boldref}[1]{\textbf{\autoref{#1}}}

\newcommand{\marskalkrefinement}[2]{\textbf{#1} #2\\}

\newtheorem{theorem}{Theorem}[section]

\newtheoremstyle{examplestyle}
  {0pt}
  {2pt}
  {\it}
  {0pt}
  {\bf}
  {}
  { }
  {\thmname{#1}\thmnumber{ #2}:\thmnote{ #3}}
\theoremstyle{examplestyle}  
\newtheorem{example}{Example}[section]

\newtheoremstyle{mystyle}
  {0pt}
  {2pt}
  {\normalfont}
  {0pt}
  {\it}
  {}
  { }
  {\thmname{#1}\thmnumber{ #2}:\thmnote{ #3}}
\theoremstyle{mystyle}  

\newtheorem{observation}{Observation}[section]

\newtheoremstyle{remarkproofstyle}
  {0pt}
  {0pt}
  {\normalfont}
  {0pt}
  {\it}
  {}
  { }
  {\thmname{#1}\thmnumber{ #2}:\thmnote{ #3}}
\theoremstyle{remarkproofstyle}

%% file: abstract.tex
\begin{abstract}
A bit layout is a sequence of fields of certain bit lengths that specifies how to interpret a serial stream, e.g., the MP3 audio format.
A layout with variable length fields needs to include meta-information to help the parser interpret unambiguously the rest of the stream; e.g. a field providing the length of a following variable length field. If no such information is available, then the layout is ambiguous.
I present a linear-time algorithm to determine whether a layout is ambiguous or not by modelling the behaviour of a serial \reader reading the stream as forward chaining reasoning on a collection of Horn clauses.
\end{abstract}

%% file: intro.tex
\section{Introduction}\label{sec:intro}
Programs can read data from files or network interfaces in serial form.
Data might not be available at once, or its consumption might be inherently sequential, such as in digital music.
Programs decode the data stream interpreting it through a structure that defines the layout, or binary format, of the bits within the stream.
Call this process \emph{deserialization}, or \emph{parsing}, or \emph{unmarshalling} interchangeably.
Among the reasons to use ad-hoc binary formats there are: \begin{inparaenum}[\itshape a)]
\item Conciseness over verboseness of \ascii-based exchange formats like \xml or \json;
\item interfacing to legacy or closed-source software that uses custom binary formats for which no parser is accessible;
\item application specific constraints on the binary format.
\end{inparaenum}
The most painful drawback of an ad-hoc binary format is its mainteinance.
Any change in the layout means changing the marshalling/unmarshalling routines, which is error prone due to the fact that bitwise logical and shifting operations are involved and off-by-one errors might fester.
It is therefore appealing to have such routines automatically derived from an high level layout specification.
Consider the Portable Network Graphic (\png) format \cite{PNG}.
\png image files are composed of a fixed 8 bytes header, followed by an arbitrary number of chunks\footnote{Terminology taken from \cite{PNG}.}.
\begin{table}[h]
\centering
\caption{The \png chunk layout.} 
\padparsergraph
\begin{tabular}{l|*{6}{c}}
\textbf{Meaning\ } & \ & \pngLength & \pngType & \pngData & \pngCRC \\\\
\hline\\
\textbf{Bytes\ }& \  & 4 & 4 & Variable & 4 \\
\end{tabular}
\label{bit:png}
\end{table}
A chunk itself has a variable length; \autoref{bit:png} shows its layout.
The \pngLength field's value tells the length of the \pngData field.
Without knowing the value of the \pngLength field it is impossible to unambiguously parse the rest of the chunk (and any following chunks).
This means that not only the presence, but also the position in the stream of the \pngLength field is crucial for deserialization.
Placing \pngLength \emph{after} \pngData would prevent deserialization since it is not possible to know at what point of the stream \pngLength would begin.
The example shows how variable fields urge the presence of meta-information in the stream. In practice these are pointer fields or terminator sequences of bits (or \emph{syncwords}).
These two solutions are not equivalent.
In the \png case it is desirable to know in advance how much memory to allocate, since the image data is buffered to be consumed e.g. for displaying on a screen.
In contrast, syncwords are preferred when the length of the variable field is not known when the data is sent.
Consider the playback of audio-streams like MP3 \cite{isoMPthree}.
 Header packets give the information about the the bitrate of the following data; this fixes the amount of buffering needed, since the data is discarded as soon as it is played back.
The arrival of a new packet header is signalled by a syncword.
\marskalkrefinement{Contributions.}{In this paper I show \emph{a method to formally check if a layout is successfully deserializable or not},
  by defining a stream model and a \emph{\reader} model, 
  that is, 
  a first-order logic axiomatization that encodes in horn clauses the behaviour of a \reader that reads and interprets a sequential stream with respect to a layout.
  I use known reasoning techniques to infer the layout properties.}
\marskalkrefinement{Paper structure.}{In \boldref{sec:background} I summarize the needed background about reasoning on knowledge bases. In \boldref{sec:method:deser} I elaborate a simplified model of the layout, and introduce the \reader model and the deserialization check. In \boldref{sec:refinements} I make the model more expressive and analyze the consequences on deserialization.
In \boldref{sec:related} I discuss related work. In \boldref{sec:marskalk} I conclude the paper and illustrate future work.}
\marskalkrefinement{Disambiguation.}{In the following, the intended meaning of the word model is ``mathematical description of a process" and \emph{not} ``interpretation that makes true a theory in first-order logic".}

%% file: background.tex
\section{Background: Knowledge Representation and First-Order Logic}\label{sec:background}
Let $\KB$ be a finite conjunction of first-order formulae of the form $\premise \Rightarrow \consequence$, called \emph{rules}.
A rule of the form $\true \Rightarrow \consequence$ is called a \emph{fact}.
$\premise$ is a conjunction of (possibly negated) predicates; $\consequence$ is a predicate\footnote{Propositions are considered here as nullary predicates.}.
$\KB$ is called \emph{knowledge base} and represents the known causal relations and facts about a modelled domain.
It is possible to infer new facts using a \emph{forward chaining} algorithm \cite{RussellAI}, that is, repeatedly applying modus ponens to the rules and facts present in $KB$ until no new facts are inferred.
A rule with no negated premises is called Horn rule (or clause); a knowledge base made of Horn rules is a Horn knowledge base.
Such class of knowledge bases is relevant because inference can be efficient \cite{dowling1984linear}.
Using forward chaining on a $KB$ of a first-order language without functions is guaranteed to always terminate, because the number of facts that can be generated is finite;
without functions no other references to domain elements than the ones explicitly mentioned in the knowledge base can be created.
Forward chaining might not terminate if the language has functions, as it might endlessly generate new facts; e.g., applying forward chaining to the Peano axioms.
When used in rule-based languages and systems \cite{clips,jess}  such as \clips or \jess, forward chaining models the reasoning process of an agent, where the knowledge base represents what the agent knows at a particular point of the reasoning.
In this case, forward chaining uses negation as failure \cite{negationAsFailure} besides modus ponens, which practically means that the lack of a fact in $\KB$ implies its falsity.
For instance, consider $\KB' = \{\neg A(1) \Rightarrow B(1), B(x) \Rightarrow A(x)\}$.
An expert system like \clips would infer $\KB'' = \KB' \cup \{B(1), A(1)\}$; modus ponens alone would not apply.

%% file: deserializability.tex
\section{Deserialization of Binary Layouts}\label{sec:method:deser}\label{sub:layout}
A layout is the sequence of fields, left to right, expected in reading a stream.
I first give a formal model to describe layouts with sufficient detail.
I then formulate a first-order formal system
$\formalSystem = \langle \readerFol{\ }, \modusponens, \axiomset \rangle$
where $\readerFol{\ }$ is a first-order language,
$\modusponens$ is the modus ponens inference rule,
$\axiomset$ a set of axioms describing the \reader's knowledge.
I will write $\axiomset \modusponens^* \alpha$, where $\alpha$ is a formula in $\readerFol{\ }$, to mean that a proof exists for $\alpha$ in $\formalSystem$.
Let the following be: $\pointers = \Pointers$;
$\items = \pointers \cup \{\field, \varfield\}$;
$\LabeledItems = \items \times \naturals$.
I represent a pair $\alabeleditem{\anitem, i} \in \LabeledItems$ as $\anitem_i$ instead of the usual tuple format.
Let $A$ be any set.
Let
$\concatenate : A^n \times A^m \longmapsto A^{n+m}$
be a family of associative concatenation operations,
 which concatenate together two tuples:
e.g. $\concatenateApp{\field_1\field_2}{\varfield_5\field_2} = \field_1\field_2\varfield_5\field_2$.
Define the family of size functions $|x| : A^k \longmapsto \naturals$ as $|x| = k$, which tells the number of elements in the tuple:
e.g., $|\field_1\field_2\varfield_5\field_2| = 4$.
The set of tuples of $A$ of any size is denoted by $A^* = \bigcup_{k \in \naturals}{A^k}$.
Given an $\vec{a} \in A^*$ and $\alpha \in A$ I will write
$\alpha \occurs \vec{a}$ meaning that $a$ occurs in $\vec{a}$ at any position;
$\alpha \occurs_i \vec{a}$ where  $\vec{a}\in A^n$ and $i\in \naturals, 0\leq i<n$, to mean that $\alpha$ occurs in $\vec{a}$ at position $i$.

Define the function $\labelfun: \items^* \longmapsto \LabeledItems^*$ as
$\labelfun(\emptyLayout) = \emptyLayout$, 
$\labelfun(\anitem) = \anitem_0$, 
$\labelfun(\concatenateApp{\anotheritem}{\anitem}) = \concatenateApp{\labelfun(\anotheritem)}{\anitem_k}$
where $\anitem \in \items, \anotheritem \in \items^{k-1}$, and $\emptyLayout \in \items^0$ is the identity element of the $\concatenate$ operator.
Let $\Layouts' = \Image(\labelfun)$.
Given a $\alayout \in \Layouts'$ I will write
$\anitem_\alabel \occurs \alayout$ meaning that $\anitem_\alabel$ occurs in $\alayout$; $\anitem \occurs \alayout$ meaning that there is an $\anitem$ occurring in $\alayout$ with any label.
The set of layouts $\Layouts \subset \Layouts'$ is defined as follows:
 $\Layouts = \{\alayout \in \Layouts' |
  \forall \anitem.\big(\anitem \occurs \alayout\big) \wedge \big(\anitem = \pointer{o}{s}_k\big) \Rightarrow
     \big(o<|\alayout|\big) \wedge \big(s\leq|\alayout|-o\big)
  \}$.
The above cryptic formal introduction is to set a framework for describing, later in the paper, extensions to the layout model:
the function $\labelfun(x)$ assigns unique labels that identify the items in a tuple $x$ with their position in $x$.
I will sometimes drop the labeling subscripts for readability.

Each layout field has a length, that is, the number of contiguous bits that will represent the content of the field in the stream.
The concrete value of the length is not important in this work.
The meaning of each $\anitem \in \items$ is defined as follows: \begin{inparaenum}[\itshape 1)]
\item $\field$ is a fixed length field;
\item $\varfield$ is a variable length field, or \emph{varfield};
\item any $\pointer{\landingpad}{\pointerspan} \in \pointers$ indicates a fixed length pointer field, where
$\landingpad$ is the offset label of the pointer, \emph{offset} in short, and
$\alabel = \landingpad + \pointerspan$ is the label of the item the pointer is pointing to --- thus I call $\pointerspan$ the \emph{span} (\emph{and not length!}) of a pointer;
I define a function
$\ptrRange : \pointers \longmapsto \naturals$
as $\ptrRange(\pointer{\landingpad}{ \pointerspan}) = \{j \in \naturals | \landingpad \leq j < \landingpad+\pointerspan\}$
which tells the \emph{range} of a pointer.
Note that the definition of $\Layouts$ rules out pointers pointing or spanning beyond $|\alayout|$.
A layout $\alayout \in \Layouts$ defines the structure of a set of concrete bitstrings, denoted by $\streams(\alayout)$.
\end{inparaenum}
\begin{example}The representation of the \png chunk of \autoref{sec:intro} is
$\pointer{2}{1}_0\field_{\ }\varfield_2\field_{ }$; and $\ptrRange(\pointer{2}{1}) = \{2\}.$
\end{example}
\subsection{\Reader Model: Axioms and Knowledge Representation}\label{subsec:semantic}
A \reader reads a stream sequentially and interprets the fields according to their layout $\alayout \in \Layouts$.
I assume that it is not possible to know whether the stream is over or not, e.g. with an \emph{end-of-stream} signal, event, or symbol.
Consider a first-order language
$\readerFol{n} = \langle \synConstants, \synVars, \synFuncs, \synPreds \rangle$
with a set of constants $\synConstants = \{c_0, \ldots, c_n\}$,
an infinite supply of variables $\synVars$,
a single binary function $\mathbf{+} \in \synFuncs$,
the unary predicate set $\synPreds_u = \{ \beg{}, \len{}, \val{} \}$,
and the ternary predicate set $\synPreds_t =\{\pspan{}\}$.
Let $\synPreds = \synPreds_u \cup \synPreds_t$.
\label{subsub:semantic}
I define the \reader model as a theory $\axiomset$ in $\readerFol{n}$, in the following way\footnote{It is understood that the theory is the conjunction of the formulae it contains.}.
Let the domain be $\naturals$.
I impose that the interpretation $\interpr{c} : \synConstants \longmapsto \naturals$ of any $c \in \synConstants$ is fixed: $\interpr{c_0}= 0, \ldots, \interpr{c_n} = n$.
To force the interpretation of the $+$ function I add to $\axiomset$ the axioms\footnote{Not reported here.} defining the addition over $\naturals$.
All $p \in \synPreds$ have a corresponding predicate, $\interpr{p}$, with integer arguments.
My intention is to give the following meanings to the predicates. Let $\alayout \in \Layouts, i \in \synVars, \anitem \in \items$; let $\anitem_{\interpr{i}} \occurs \alayout$. Then
\begin{inparaenum}[\itshape a)]
\item $\beg{i}$ means ``the \reader knows where $\anitem_{\interpr{i}}$ begins in the stream";
\item $\len{i}$ means ``the \reader knows $\anitem_{\interpr{i}}$'s length";
\item $\val{i}$ means ``the \reader knows $\anitem_{\interpr{i}}$'s content";
\item $\pspan{o,s,i}$ tells that there is a pointer field with label $\interpr{i}$ that contains a measure of how many bits there are between the beginning of the fields labeled with $\interpr{o}$ and $\interpr{o+s}$. Less verbosely, it means $\anitem_{\interpr{i}} = \pointer{\interpr{o}}{\interpr{s}}_{\interpr{i}}$. 
\end{inparaenum}
I define the behaviour of a \reader with the following axioms $\axiomset_{\alayout}$
 (implicitly universally quantified):
\begin{enumerate}
\item \itshape{The \reader knows where $\myanitemarg{0}$ begins.}
     \begin{equation}\tagBegin\label{ax:begin}
       \true \Rightarrow \beg{0}
     \end{equation}
\item \itshape{If a \reader knows where $\myanitem$ begins and its length, then it knows its value and where $\myanitemarg{i+1}$ begins.}
      \begin{equation}\tagForward\label{ax:forward}
       \beg{i} \wedge \len{i} \Rightarrow \val{i} \wedge \beg{i+1}
      \end{equation}
\item \itshape{If a \reader knows where $\myanitemarg{i+1}$ begins and the length if its predecessor $\myanitem$, then it knows where $\myanitem$ begins and its value.}
      \begin{equation}\tagBackward\label{ax:backward}
        \beg{i+1} \wedge \len{i} \Rightarrow \beg{i} \wedge \val{i}
      \end{equation}
\item \itshape{If a \reader knows where $\myanitem$ and its successor $\myanitemarg{i+1}$ begin, then it knows $\myanitem$'s length.}
            \begin{equation}\tagJoin\label{ax:join}
              \beg{i} \wedge \beg{i+1} \Rightarrow \len{i}
            \end{equation}
\item \itshape{If $\myanitem = \pointer{\interpr{o}}{\interpr{b}}$ and the \reader knows \begin{inparaenum}[\bf a)]
      \item the value of $\myanitem$
      \item where $\myanitemarg{o}$ begins,
      \end{inparaenum} then it knows where $\myanitemarg{o+b}$ begins.}
      \begin{equation}\tagJumpRight\label{ax:jumpRight}
        \pspan{o,s,i} \wedge \val{i} \wedge \beg{o} \Rightarrow \beg{o+s}
      \end{equation}
\item \itshape{If $\myanitem = \pointer{\interpr{o}}{\interpr{b}}$ and the \reader knows \begin{inparaenum}[\bf a)]
      \item the value of $\myanitem$
      \item where $\myanitemarg{o+b}$ begins,
      \end{inparaenum} then it knows where $\myanitemarg{o}$ begins.}
      \begin{equation}\tagJumpLeft\label{ax:jumpLeft}
        \pspan{o,s,i} \wedge \val{i} \wedge \beg{o+s} \Rightarrow \beg{o}
      \end{equation}
\end{enumerate}
\normalfont
The above axioms are common to all layouts; the following are axioms that are added according to the specific shape of the layout $\alayout$ under analysis\footnote{For completeness one can extend the layout model with a \emph{constant field} $\const$
                                                                                                                                                            to signal the end of a variable field with a constant pattern.
                                                                                                                                                               Thus, for each $\const_i \occurs \alayout$, add an axiom $\beg{i}$.
                                                                                                                                                               This is sound under the assumption that the bits in the stream before $\const_i$ are such that the interpretation is not ambiguous,
                                                                                                                                                               e.g. the pattern in $\const_i$ occurs in the bits of $\varfield_{i-1}$.}.
For each $\myanitem \occurs \alayout$: if $\anitem = \field$ or $\anitem = \pointer{\interpr{\landingpad}}{\interpr{\pointerspan}}$, then the \reader knows
\begin{equation}\tagFieldAx\label{ax:fieldAx}
\true \Rightarrow \len{i}
\end{equation}
and additionally if $\anitem = \pointer{\interpr{o}}{\interpr{s}}$ the \reader knows
\begin{equation}\tagPtrAx \label{ax:ptrAx}
\true \Rightarrow \pspan{\landingpad, \pointerspan, i}
\end{equation} 
I will drop the subscript to $\axiomset_{\alayout}$ whenever the $\alayout$ it refers to is clear from the context. I will call $\axiomset$ the \emph{initial knowledge base}.
In the following I will abuse the notation by using the same digit symbols to represent both \begin{inparaenum}[\itshape a)]
\item the value represented
\item the syntactic entity representing it,
\end{inparaenum} therefore not explicitly representing the interpretation function $\interpr{\cdot}$ when such distinction is not necessary.
Wrapping up, for each $\alayout \in \Layouts$
there is a formal system
$\formalSystem_{\alayout}= \tuple{\readerFol{|\alayout|}, \modusponens, \axiomset}$ which is $\alayout$'s \reader model;
and in the following, whenever i write about a \reader, I implicitly refer to such a structure.
\subsection{Ambiguity}
The presence of a varfield creates ambiguity.
For instance the layout $\field\varfield$ is ambiguous,
because there is no way for a \reader to know $\varfield_{1}$'s length;
likewise in $\field\varfield\field$ there is no way, in a concrete stream, to tell $\varfield_1$ from $\field_2$.
 A layout $\alayout$ is unambiguous, or \emph{deserializable}, if and only if a \reader can infer the lengths of all $\varfield_i \occurs \alayout$.
Pointers are \emph{bounding} items, in that their presence can bound varfields and therefore disambiguate the layout.
\begin{example}
Consider layout $\pointer{1}{2} \varfield_1 \pointer{1}{1} \pointer{2}{1}_3$.
Item $\varfield_1$ is bounded by $\pointer{1}{2}_0$ and $\pointer{1}{1}_2$; it is not bounded by $\pointer{2}{1}_3$.
\end{example}
\begin{theorem}\label{the:necessarybounding}
\textbf{Necessary condition for deserializability.}
If $\alayout \in \Layouts$ is deserializable then for all $\varfield_j \occurs \alayout$ there exists a pointer $x = \pointer{b}{s}_p \occurs \alayout$ such that $j \in \ptrRange(x)$.
\end{theorem}
\begin{proof}
Let $\varfield_j \occurs \alayout$.
 Since $\alayout$ is deserializable, it is true that $\axiomset \modusponens^* \len{k}, 0 \leq k \leq |\alayout|$.
I show that any proof of $\len{j}$ contains the application of \ref{ax:jumpRight} or \ref{ax:jumpLeft}, with $o\leq j <o+s$, by applying the inference steps backwards.
Observe that
\begin{inparaenum}[\itshape 1)]
\item $\len{j} \notin \axiomset$ , otherwise $\varfield_j$ would not be a varfield.
\item The \refJoin{j} axiom is the only axiom that allows to infer $\len{j}$, so any proof necessarily applies \refJoin{j}.
\item $\varfield_{0} \noccurs \alayout$, otherwise $\alayout$ would not be deserializable.
\end{inparaenum}
Assume that there are no proofs involving \ref{ax:jumpRight} or \ref{ax:jumpLeft}. Further, observe the premises of \refJoin{j}: $\beg{j}$ cannot be inferred through \refBackward{j}, since $\len{j}$ is in its premises leading to circularity; likewise $\beg{
j+1}$ cannot be inferred through \refForward{j}. Thus,
\begin{inparaenum}[\itshape 1)]
    \item \label{step:foundation}$\beg{j}$ must then be inferred through \refForward{j}, which means that a proof is a chain of \refForward{k}, with $0\leq k\leq j$;
    \item consequently $\beg{j+1}$ is inferred through \refBackward{j+1}.
     Since the layout is finite, at most $|\alayout|-j$ \refBackward{k} inferences can be done.
 Note that any inference of a new $\beg{k}$ depends on some other $\beg{l}$;
  the only axiom of such shape is $\beg{0}$;
  so applying only \refBackward{k} will not close the proof, which contradicts the hypothesis that $\axiomset \modusponens^* \len{k}$.
  Hence, the proof of at least one of $\beg{t}$ with $t>j$ must include an application of
   \refJumpRight{l,s,p}
   where $l+s = t$, because it allows to infer $\beg{t}$ from $\beg{l}$ where $l < j < t$;
    and as seen at point \ref{step:foundation}, leads backwards to the axiom $\beg{0}$. This requires that $\pointer{l}{s}_p\occurs \alayout$.
 \end{inparaenum}
\end{proof}

 Note that this condition is not sufficient: layout $\pointer{0}{4}\varfield_1\pointer{3}{1}\varfield_3$ satisfies the necessary condition, but it is not possible to know the length of $\varfield_1$ nor $\varfield_3$.
The interesting fact about \autoref{the:necessarybounding} is that there is no constraint on the value $p$ in the layout. This means that pointer and varfield can be in any relative order.
\begin{example}\label{ex:convoluted}
Consider
$\alayout =  \field_{\ } \pointer{2}{3}_{\ } \field_{\ } \varfield_3 \varfield_4 \pointer{3}{1}$ and \autoref{algpic:necessary}.
A parser can read until $\field_{2}$ by applying \refForward{0},
storing the value of the pointer field at 1;
from that point on it can buffer the stream (\refJumpRight{2,3,1}) until $\pointer{3}{1}_{5}$,
which once read with \refForward{5}
allows to determine the lengths of $\varfield_3$ and $\varfield_4$ through \refJumpLeft{3,1,5} and \ref{ax:join},
where $i \in \{2,3\}$ followed by, respectively, \refForward{3} and \refBackward{3}.
\begin{figure}[h]
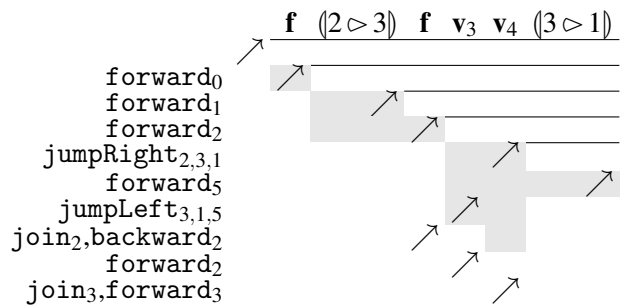

\centering
\padparsergraph
\begin{tabular}{r*{7}{c}}
&&$\field$ &$\pointer{2}{3}$& $\field$ &$\varfield_3$ &$\varfield_4$ &$\pointer{3}{1}$\\\cline{3-8}
&\headpostab&&&&&\\\cline{4-8}
\refForward{0}&&\headpostabbuf&&&&\\\cline{5-8}
\refForward{1}&&&\headpostabbuf&&&\\\cline{6-8}
\refForward{2}&&&\bufferized&\headpostabbuf&&\\\cline{8-8}
\refJumpRight{2,3,1}&&&&&\bufferized&\headpostabbuf\\
\refForward{5}&&&&&\bufferized&\bufferized&\headpostabbuf\\
\refJumpLeft{3,1,5}&&&&&\headpostabbuf&\bufferized&\\
\refJoin{2},\refBackward{2}&&&&\headpostab&&\bufferized&\\
\refForward{2}&&&&&\headpostab&&\\
\refJoin{3},\refForward{3}&&&&&&\headpostab&
\end{tabular}
\caption{Parsing a stream, \autoref{ex:convoluted}. Read top to bottom. Each row is a snapshot of the state of the parsing before the application of the axiom on the left.
The arrow indicates the position of the parser in the stream;
the greyed areas represent pictorially the amount of buffering.
The thin line represents the amount of stream consumed.
The deserialization is successful when all the stream is consumed, and no buffering is left.
}
\label{algpic:necessary}
\end{figure}
\end{example}
\subsection{Deserializability Check Algorithm}
\autoref{alg:simpleDeser} shows in pseudocode how to check for deserializability. It is correct and complete because $\forwardChainingInfer$ is \cite{RussellAI}\footnote{$\forwardChainingInfer$ can be replaced with any existing implementation of forward chaining.}. 
\begin{observation}[\autoref{alg:simpleDeser} terminates.]
Let $\alayout \in \Layouts$ be \autoref{alg:simpleDeser}'s input.
The \ref{ax:forward} and \ref{ax:jumpRight} are the only axioms that can introduce more complex terms using the $+$ function.
I show that such axioms are applied at most $|\alayout|$ times to produce new facts.
 \ref{ax:jumpRight} cannot introduce more facts than the number of pointer fields $\pointerFieldNum < |\alayout|$ since no inference can introduce new $\pspan{o,s,i}$;
 \ref{ax:forward} introduces a new $\beg{i+1}$ if there is a $\len{i}$ fact in the knowledge base; there are two cases.
 \begin{inparaenum}[\itshape 1)]
 \item $\len{i}$ was already present in the initial knowledge base, so $i < |\alayout|$, thus $i+1 \leq |\alayout|$.
 \item  $\len{i}$ could have been inferred through \ref{ax:join}, but this breaks the assumption that $\beg{i+1}$ is not in the knowledge base.
 \end{inparaenum}
The conditional checks the results of the forward chaining algorithm by comparing two finite structures. This proves termination.
\end{observation}

The algorithm is $O(|\alayout|)$ because propositionalizing the axioms takes linear time due to the shape of the axioms and the absence of uninterpreted function symbols;
then, inference is linear on propositional Horn knowledge bases \cite{dowling1984linear,RussellAI}.
\begin{figure}[b]
\begin{algorithm}[H]
 \KwData{$\alayout \in \Layouts$, $V = \{i | \varfield_i \occurs \alayout\}$ }
 \KwResult{A modified knowledge base $\axiomset'$ and the inference graph $\infGraph$}
 Build $\axiomset$ according to \autoref{subsec:semantic}\;
 $\simpleSuccessStruct \leftarrow \forwardChainingInfer(\axiomset)$\;
 \eIf{$\exists . i \in V | \len{i} \notin \axiomset' $}{\Return $\simpleFailureStruct$\;}{
 \Return $\simpleFailureStructRet$
 }
 \caption{Deserializability check}
 \label{alg:simpleDeser}
\end{algorithm}
\end{figure}
Note that the deserializability check is not sufficient to decide properties of $\streams(\alayout)$, e.g. whether a layout $\alayout$ has $\streams(\alayout) = \emptyset$.
      Consider the following scenario: $\alayout = \pointer{0}{3}\field_{1}\varfield_{2}\field_{3}$,
       and suppose the lengths of the non-variable length fields are, respectively, 1 bit, 3 bits and 3 bits.
      If the meaning of the value of the pointer is to measure the number of bits of the items with labels in $\ptrRange(\pointer{0}{3})$, this value cannot encode that number with just one bit.
        This additional check is not needed to decide deserializability, and can be performed after the deserializability check by analysing, considering field lengths and encodings, the spans of all pointers.     
  

%% file: refinements.tex
\section{The Repetition Field}\label{sec:refinements}
A reasonable extension of the model is to have variable occurrences of portions of layout, like the Kleene star in regular expressions:
 $\alayout' = \pointer{1}{1}_0\repetition{\pointer{\intlist{1,0}}{1}_{\intlist{1,0}}\varfield_{\intlist{1,1}}}_1$
 to indicate the infinite set of layouts beginning with a pointer and a sequence of alternating pointer and variable fields.
 The layout is now a tree structure with a new basic item $\repetition{\xspace}$ called \emph{repetition}, and identifiers are tuples of integers in $\naturals^*$.
 Let $\pointers' = \EnhancedPointers$. Let the set $\repeitems$ be defined recursively as follows:
  $\emptyLayout \in  \repeitems$,
  $\anitem \in  \repeitems$ where $\anitem \in \items'$,
  $\repetition{W} \in  \repeitems $ where $W \in \repeitems^*$. 
 Nothing else is in $\repeitems$.
 The empty layout $\emptyLayout$ is defined as the identity operator for $\concatenate$, and
 $\items' = \items \cup \pointers'$.
 To identify each item I define, following the same pattern of \autoref{sub:layout}, the set
 $\LabeledRepItems = \repeitems \times \naturals^*$ (cf. $\LabeledItems$)
 and the function
 $\labelfunrep : \repeitems^* \longmapsto \LabeledRepItems^*$ (cf. $\labelfun$)
 defined as follows:
 
\begin{table}[h]
  \centering
 \begin{tabular}{lp{2cm}r}
       {$\!\begin{aligned} 
                \labelfunrep(\emptyLayout) & = \emptyLayout\\
                 \labelfunrep(\anitem) & = \labelfunrep_0(\anitem)\\
                 \labelfunrep_l(\emptyLayout) & = \emptyLayout\\
                  \labelfunrep_l(\anitem)
                     & = \anitem_l \qquad \mbox{where~} \anitem \neq \repetition{\anotheritem}\\
                   \labelfunrep_l(\repetition{\anitem})
                     & = \repetition{\labelfunrep^l_0(\anitem)}_l\\
                   \labelfunrep_l(\concatenateApp{\anotheritem}{\anitem})
                     & = \concatenateApp{\labelfunrep_l(\anotheritem)}{\labelfunrep_{l+n}(\anitem)}\\ 
              \end{aligned}$}
              & &
        {$\!\begin{aligned} 
                         \labelfunrep_l^k(\emptyLayout) & = \emptyLayout\\
                         \labelfunrep_l^k(\anitem) & =
                            \anitem_{\concatenateApp{k}{l}}\\
                         \labelfunrep_l^k(\anitem) & =
                            \repetition{\labelfunrep_0^{\concatenateApp{k}{l}}(\anitem)}_{\concatenateApp{k}{l}}\\
                         \labelfunrep_l^k(\concatenateApp{\anotheritem}{\anitem}) & =
                            \concatenateApp{\labelfunrep^k_l(\anotheritem)}{\labelfunrep^k_{l+n}(\anitem)}
                      \end{aligned}$}
 \end{tabular}
 \end{table}
 In words: the function $\labelfunrep_l$ labels the items left to right, starting from $l \in \naturals$ and introducing a context $l$ when it applied to a repetition; $\labelfunrep_l^k$ labels the items starting from $l \in \naturals$, in the context $k \in \naturals^*$.

Let $\RepeatLayouts' = \Image(\labelfunrep)$.
 The set of layouts $\RepeatLayouts \subset \RepeatLayouts'$ is defined as follows:
 \begin{align}
 \RepeatLayouts = &\{\alayout \in \RepeatLayouts' |\nonumber\\
    &
    \forall \anitem . \bigg( \anitem = \pointer{a}{b}_k \occurs \alayout \bigg) \Rightarrow  \bigg(
          \Big( \anitem  = \pointer{\concatenateApp{s}{c}}{b}_{\concatenateApp{s}{d}}  \vee
          \anitem  = \pointer{s}{b}_{\concatenateApp{s}{e}}\Big)
            \bigg),\label{set:nesting}\\
    &
    \forall y. \bigg( \repetition{y}_m \occurs \alayout\bigg) \Rightarrow \bigg( \Big( \pointer{\concatenateApp{m}{f}}{g} \occurs \alayout \Big)
    \Rightarrow g \leq |l|-f
    \bigg),
    \label{set:existinga}\\
    &
    \forall \anitem . \bigg( \anitem = \pointer{f'}{g'}_{c'} \occurs \alayout  \bigg)
        \Rightarrow \Big(  g' \leq |\alayout|-c'\Big)
        ,
        \label{set:existingb}\\
    &l \in \repeitems^*, y = \labelfunrep_0^m(l), \{b,c,c',d,e,f,f',g,g'\} \subset \naturals, \{a,k,m\} \subset \naturals^*\setminus \naturals^0, s\in \naturals^*\nonumber\\
 \nonumber\}.
 \end{align}
  In words, in all $\alayout \in \RepeatLayouts$:
   (\ref{set:nesting}) tells that the offset of any pointer refers to a label of a parent scope, or to an element at the same level.
           This is needed to prevent ambiguous references.
           For instance in
           $\alayout = \pointer{\intlist{1,0}}{1}_0\repetition{\field_{\intlist{1,0}}}_1 \notin \RepeatLayouts$
           the pointer $\pointer{\intlist{1,0}}{1}_0$ points to  $\field_{\intlist{1,0}}$ which in a concrete stream can appear an unbounded number of times and therefore the pointer would be ambiguous.
   (\ref{set:existinga}) and (\ref{set:existingb}) tell that all pointers have spans that do not exceed the number of fields of the context they are in.
  The list labels give the context needed to state this property.
  The mapping $\treetoflat : \RepeatLayouts \longmapsto 2^\Layouts$ maps, informally\footnote{A formal definition is omitted. The mapping must take care of \begin{inparaenum}[\itshape a)]
  \item flattening the label structure
  \item change the pointer elements' spans and offsets.
  \end{inparaenum} I rely on the intuitive meaning of $\treetoflat$ to avoid a complicated formal definition.}, to a set of $\alayout \in \RepeatLayouts$ without repetitions corresponding to all the combinations of unwindings of the repetitions, $0,1,2 \ldots$ times.
  For the above example: $\treetoflat(\alayout) = \{\pointer{0}{0},
  \pointer{0}{2}\pointer{1}{1}\varfield,
  \pointer{0}{2}\pointer{1}{1}\varfield\pointer{3}{1}\varfield \ldots \}$.
  \subsection{Parser Model}\label{subsub:parserrep}
  As in \autoref{sec:method:deser}, I will define a formal system
  $\formalSystem_{\alayout} = \tuple{\readerRepFol{|\alayout|}, \modusponens, \axiomset}$,
  where $\alayout \in \RepeatLayouts$, to analyze the deserializability of $\alayout$.
 The system's first-order language is $\readerRepFol{n} = \langle \synConstants, \synVars, \synFuncs', \synPreds' \rangle$
 where $\synFuncs' = \synFuncs \cup \{\concatenate\}$
 and $\synPreds' = \synPreds \cup \{\rep{}, \repLen{}\}$ where $\rep{}$ is a binary predicate and $\repLen{}$ is a unary predicate.
  Let the domain be $\naturals^*$;
 predicate symbols in $\synPreds'$ map to predicates of the same arities and names.
 The interpretation of constant symbols is fixed as explained in \autoref{subsub:semantic}, \emph{mutatis mutandis}.
 The function symbol $+$ is interpreted as addition over integers;
 it is left undefined for arguments $a \notin \naturals$.
 The symbol $\concatenate$ corresponds to the tuple concatenation function introduced in \autoref{sub:layout}\footnote{Axioms defining the behaviour of integers, lists of integers and the relevant operations are not reported here.}.
 The axioms of \autoref{subsub:semantic} are lifted to the list domain:
\begin{align}
       \true &\Rightarrow \beg{0}\tagBegin\label{ax:begin:ref}\\
       \beg{\intlist{b,a}} \wedge \len{\intlist{b,a}}
          &\Rightarrow 
          \val{\intlist{b,a}} \wedge \beg{\intlist{b,a+1}}\tagRepForward\label{ax:forward:ref}\\
       \beg{\intlist{b,a+1}} \wedge \len{\intlist{b,a}}
          &\Rightarrow \beg{\intlist{b,a}} \wedge \val{\intlist{b,a}}\tagRepBackward\label{ax:backward:ref}\\
       \beg{\intlist{b,a}} \wedge \beg{\intlist{b,a+1}}
          &\Rightarrow \len{\intlist{b,a}}\tagRepJoin\label{ax:join:ref}\\
       \pspan{\intlist{b,a},s,i} \wedge \val{i} \wedge \beg{\intlist{b,a}} &\Rightarrow \beg{\intlist{b,a+s}}\tagRepJumpRight\label{ax:jumpRight:ref}\\
       \pspan{\intlist{b,a},s,i} \wedge \val{i} \wedge \beg{\intlist{b,a+s}} &\Rightarrow \beg{\intlist{b,a}}\tagRepJumpLeft\label{ax:jumpLeft:ref}
\end{align}
where $a,s \in \naturals, b,i \in \naturals^*$ and $+$ has higher precedence than $\concatenate$.
The intended meaning of $\rep{a,l}$ is
``there is a repetition at position $a$ which contains $l$ fields".
Note that a repetition is a field, consistently with how repetitions are labeled.
$\repLen{a}$ means ``the \reader knows the length of the repetition at position $a$". 
  \begin{align}
          \rep{\intlist{b,a},l} \wedge \beg{\intlist{b,a}} \wedge \beg{\intlist{b,a+1}} &\Rightarrow \repLen{\intlist{b,a}} \tagRepLenAx\label{ax:replen:ref}\nonumber\\
          \rep{\intlist{b,a},l} \wedge \beg{\intlist{b,a}}                              &\Rightarrow \beg{\intlist{b,a,0}}  \tagRepLenHead\label{ax:replen:accesshead:ref}\nonumber\\
          \rep{\intlist{b,a},l} \wedge \beg{\intlist{b,a+1}}                            &\Rightarrow \beg{\intlist{b,a,l}}\tagRepLenTail\label{ax:replen:accesstail:ref}\nonumber
  \end{align}
  where $a,l \in \naturals$ and $b \in \naturals^*$. Axiom \ref{ax:replen:ref} tells how a \reader gets to know the length of a repetition;
\ref{ax:replen:accesshead:ref} and \ref{ax:replen:accesstail:ref} tell how the \reader accesses the fields inside a repetition.
 For each $\repetition{\anitem_{\intlist{b,0}}\ldots\anitem_{\intlist{b,l-1}}}_b \occurs \alayout$, $\axiomset$ contains the facts
\begin{equation}\tagRepAx\label{ax:rep}
       \true \Rightarrow \rep{b,l}
\end{equation}
  As no axioms in $\axiomset$ allow to deduce any $\rep{}$, this is the only way they can be included in the knowledge base.
  This prevents by construction to have $l \notin \naturals$,
  without the need of typing $\readerRepFol{n}$ or defining $+$ for all $i,j \in \naturals^*$.
  Additional axioms \ref{ax:fieldAx} and \ref{ax:ptrAx} are lifted to the list domain and added to the knowledge base under the same circumstances described for their counterparts in \autoref{subsub:semantic}.
  \paragraph{Caveat!}
  Consider $\alayout = \pointer{1}{1}_0\repetition{\varfield_{\intlist{1,0}}}_1 \in \RepeatLayouts$.
  Applying \autoref{alg:simpleDeser} with the modified knowledge base tells that $\alayout$ is deserializable (\autoref{algpic:caveat}(a)).
  This is unsound: knowing the length of the repetition field does not allow to discriminate the occurrences of $\varfield$ in a concrete stream.
   The problem that the example exposes is that the theory confuses in a single identifier $\intlist{1,0}$ \emph{all} the occurrences of the varfield in the repetition.
   I illustrate how to fix this shortcoming after some preliminary definitions.
    Let $\alayout' = \reverse(\alayout)$ be the permutation of $\alayout$ defined as follows:
         \begin{align*}
          (\anitem \notin \pointers') ~\wedge ~(\anitem_i \occurs \alayout) & ~\Leftrightarrow~ \anitem_{|\alayout|-i}\occurs \alayout'\\
          \pointer{\intlist{a,k}}{b}_i \occurs \alayout & ~ \Leftrightarrow ~
                  \pointer{\concatenateApp{|\alayout|-(a+b)}{k}}{b}_{|\alayout|-i}  \occurs \alayout'
         \end{align*}
    where $a\in \naturals, k \in \naturals^*$.
  Parsing $\reverse(\alayout)$ is equivalent to parsing $\alayout$ backwards, that is, substituting the axiom $\beg{0}$ with $\beg{|\alayout|}$.
   \begin{example} Let $\alayout = \pointer{1}{2}_0\field_{1}\repetition{
                                                      \pointer{\intlist{1,0}}{2}
                                                      \varfield}_2$. Then
$\reverse(\alayout) =
                            \repetition{\pointer{\intlist{1,0}}{2}\varfield}_0
                                  \field_{1}
                                    \pointer{0}{2}_2$.
    \end{example}
  Note that $\reverse(\alayout) \in \RepeatLayouts$.
  Furthermore, let $\itemarr \in \repeitems^*$. Then $\itemarr^n$ is the structure such that
  \begin{align*}
  (\anitem \notin \pointers') ~ \wedge ~ (\anitem \occurs_i \itemarr) & ~\Leftrightarrow~  \anitem \occurs_{|\itemarr|-i} \itemarr^n\\
  \pointer{\intlist{a,k}}{b} \occurs_i \itemarr & ~\Leftrightarrow ~\pointer{\intlist{a+n,k}}{b} \occurs_i \itemarr^n
  \end{align*}

where $a,n \in \naturals, k \in \naturals^*$. The sequence $\itemarr^n$ is the same as $\itemarr$, where the head of all offset labels is increased by $n$.
   \begin{example}
    Let $\itemarr \in \repeitems^*$. Then
    \begin{align*}
    \itemarr  &= \pointer{0}{2}\repetition{
                              \pointer{\intlist{1,0}}{2}
                              \varfield
                              \repetition{
                                  \pointer{\intlist{1,2,0}}{2}\field}
                                  }\\
    \itemarr^3 &= \pointer{3}{2}\repetition{
                                  \pointer{\intlist{4,0}}{2}
                                  \varfield
                                  \repetition{
                                      \pointer{\intlist{4,2,0}}{2}\field}
                                      }
    \end{align*}

    This transformation takes care of properly translating the pointer offsets when concatenating tuples of fields, as will happen below.
   \end{example}
   As a last premise, \autoref{the:necessarybounding} is lifted to include repetitions.
   Let $\ptrRange' : \pointers' \longmapsto \naturals^*$ be defined as $\ptrRange'(\pointer{\intlist{b,a}}{s}) = \{\intlist{b,k}| a \leq k < a+s \}$ where $a \in \naturals$.
      \begin{theorem}\label{the:necessaryboundingrepeat}
      \textbf{Necessary condition for deserializability with repetitions.}
      If $\alayout \in \RepeatLayouts$ is deserializable then for all $\anitem_j \occurs \alayout$, where $\anitem \in \{\repetition{\anotheritem},\varfield\}$, then there exists a pointer $x = \pointer{b}{s}_p \occurs \alayout$ such that $j \in \ptrRange'(x)$.
      \end{theorem}
      The proof is similar to that of \autoref{the:necessarybounding} and is therefore omitted.\qed

    Observe that there are $r \in \repeitems^-\subset \repeitems^*$ such that
    $\alayout = \labelfunrep(r)$ is not deserializable, but if prepended with a bounding pointer they are:
    \begin{equation}\tag{\scName{once}}\label{eq:once}
      \alayout' = \labelfunrep(\concatenateApp{\pointer{1}{|r|}}{r^1})
    \end{equation}
     This is the case of \autoref{algpic:caveat}(c).
    This means that there is an item $\anitem_i \occurs \alayout'$ such that knowing $\beg{1}$  and $\beg{|r|+1}$ entails $\len{i}$.
    If one considers
 \begin{equation}\tag{\scName{twice}}\label{eq:twice}
    \alayout'' = \labelfunrep(\concatenateApp{\pointer{1}{2|r|}}{\concatenateApp{r^1}{r^{|r|+1}}})
 \end{equation}

  then the following can be proved true $\forall r \in \repeitems^-$, thus when $\alayout'$ is deserializable and $\alayout$ is not:

    \begin{theorem}\label{the:reverse}
    $\alayout''$ is deserializable $\Leftrightarrow$ $\alayout_r = \reverse(\labelfunrep(r))$ is deserializable.
    \end{theorem}
    \begin{proof}
    ($\Rightarrow$, \scName{sketch.}) Suppose $\alayout_r$ is not deserializable.
    This means that there exists an $\anitem_i \occurs \alayout_r$ whose length is unknown,
    which corresponds in $\alayout''$ to the two items $\anitem_{i+1}$ and $\anitem_{i+|r|+1}$.
    Observe that no pointers can span from $r^1$ to $r^{|r|+1}$ by construction, which together with \autoref{the:necessaryboundingrepeat} means that there is no chance that the deserializability of $\alayout''$ comes from concatenating $r^1$ and $r^{|r|+1}$.
    Then the knowledge of $\len{i+1}$ depends on $\beg{1}$ and $\beg{|r|+1}$,
    and that of $\len{i+|r|+1}$ depends on $\beg{|r|+1}$ and $\beg{2|r|+1}$, because $\alayout'$ is deserializable.
    $\beg{1}$ and $\beg{2|r|+1}$ can be reached from $\beg{0}$, respectively applying \refForward{0} and \refJumpRight{0,2|r|,0}.
    $\beg{|r|+1}$ can be inferred in two ways:
    \begin{inparaenum}[\itshape 1)]
    \item from $\beg{1}$ through $r^{1}$, but this contradicts that $\alayout$ is not deserializable because if one could infer $\beg{|r|+1}$ from $\beg{1}$ then $\alayout$ would be deserializable. Contradiction.
    \item from $\beg{2|r|+1}$, backwards through $r^{|r|+1}$, which would then mean that one could infer $\beg{|r|+1}$ from $\beg{2|l|+1}$, which means that $\alayout_r$ is deserializable. Contradiction.
    \end{inparaenum}\\
    ($\Leftarrow$, \scName{sketch.}) The pointer $\pointer{1}{2|r|}_0$ allows to buffer the whole layout until the end of $r^{|r|+1}$.
  Since $\alayout_r$ is deserializable, the \reader can parse backwards the whole span of $\pointer{1}{2|r|}_0$.
    \end{proof}
     \begin{example}
     Let $r = \varfield_{0}\pointer{0}{1}_{1}$. Observe that
     $\alayout_r = \reverse(\labelfunrep(r)) = \pointer{1}{1}_{0}\varfield_{1}$ is deserializable; the layout $\alayout'' = \labelfunrep(\concatenateApp{\pointer{1}{4}_0}{\concatenateApp{r^1}{r^{3}}})$ becomes
      $$\pointer{1}{4}_0\varfield_{1}\pointer{1}{1}_{2}\varfield_{3}\pointer{3}{1}_{4}$$
       and is parsed by following the first pointer and reading backwards all that was buffered, since it is possible to infer the length of the varfields.
     \end{example}
   \begin{observation}
    [$\alayout''$ is deserializable $\Rightarrow \alayout'$ is deserializable.]
  By \autoref{the:reverse} it means
  $\alayout_r = \reverse(\labelfunrep(r))$ is deserializable $\Rightarrow \alayout'$ is deserializable, which is true because \refJumpRight{1,|r|,0} allows to infer $\beg{|r|+1}$ and then read backwards since $\alayout_r$ is deserializable;
  \end{observation}
  \begin{observation} [If $\alayout''$ is deserializable, all
    $\alayout^{n} = \labelfunrep(\concatenateApp{\pointer{1}{n|r|}}
                                {\intlist{r^1,\ldots,r^{(n-1)|r|+1}}})$
   are, $n>2$.] Therefore it does not matter how many repetitions of $r$ are there, since once $\beg{n|r|+1}$ is known, the stream is reconstructed backwards.
   \end{observation}
   \begin{observation} [Otherwise, no $\alayout^{n}$ can be deserializable.]
  This corresponds to those cases where not even the reversed layout is deserializable, such as
  $\varfield_0$ or $\pointer{2}{1}\varfield_1\varfield_2$.
  Since no pointers of any $r^i$ can span beyond $r^i$,
  a \reader will not be able to proceed either forward to $\beg{|r|+1}$ or backwards from $\beg{n|r|+1}$ to $\beg{(n-1)|r|+1}$.
    \end{observation}
  One can therefore transform a layout under analysis
  $\alayout_0$ into
  $\alayout_1$ by substituting all
  $\repetition{\labelfunrep^k_0(r)}_k \occurs \alayout_0$ with $\repetition{\concatenateApp{\labelfunrep^k_0(r)}{\labelfunrep^k_{|r|}(r^{|r|})}}_k \occurs \alayout_1$:
  that is, duplicating the content of each repetition.
  Observe that such device creates the same environment described in (\ref{eq:twice}):
  when the body of a repetition $\repetition{\anotheritem}_k$ is entered with \refRepLenAx{j}, one knows the extremes of a repetition unwinded twice.
  This is, as sketched, sufficient to determine the deserializability of the repetition.
  This resumes the soundness of the model.
  The duplicating transformation is polynomial\footnote{A coarse estimate can be  $O(nk)$, where
  $n$ is the number of all items appearing in the layout, and
  $k$ is the maximum level of nesting of repetitions;
  observe that $k\leq n$ since repetitions are items too.}, terminates because of the finiteness of layouts, and preserves deserializability:
  i.e., if $\alayout_1$ is deserializable, so is $\alayout_0$\footnote{And the contrapositive: if $\alayout_0$ is not deserializable, $\alayout_1$ is not deserializable.}.
   The advantage of this solution is that it reuses the formal system defined above and does not require side proofs in the formal system $\formalSystem$.
   \autoref{alg:simpleDeser} is upgraded to \autoref{alg:complicated}.
\begin{figure}[t]
          \centering
          \begin{subfigure}[b]{0.5\textwidth}
                  \padparsergraph
                  \begin{tabular}{r*{4}{c}}
                      &&$\pointer{1}{1}_0$& $\repetition{\varfield_{\intlist{1,0}}}_1$&\\\cline{3-4}
                      &\headpostab&&&\\\cline{4-4}
                      \refForward{0}&&\headpostabbuf&&\\
                      \refJumpRight{1,1,0}&&&\headpostabbuf&\\
                      \refRepLenAx{1}&&&\headpostabbuf&  \\
                      \refJoin{\intlist{1,0}},\refBackward{\intlist{1,0}}&&\headpostab&&\\
                      \refJoin{\intlist{1}}&&\headpostab&&
                    \end{tabular}
                    \caption[]{}
                    \label{alg:caveat:analysis}
          \end{subfigure}%
          \\\vspace{0.5cm}
          \begin{subfigure}[b]{0.45\textwidth}
          \padparsergraph
          \begin{tabular}{r*{5}{c}}
                                 &&$\pointer{1}{2}_0$& $\varfield_{1}$ & $\varfield_{2}$ &\\\cline{3-5}
                                 &\headpostab&&&&\\\cline{4-5}
                                 \refForward{0}&&\headpostabbuf&&&\\
                                 \refJumpRight{1,2,0}&&&\bufferized&\headpostabbuf& \\ 
                                 &&&&&
                               \end{tabular}
                               \caption[]{}
                  \label{alg:caveat:instance}
          \end{subfigure}
          \begin{subfigure}[b]{0.45\textwidth}
          \padparsergraph
           \begin{tabular}{r*{4}{c}}
                                           &&$\pointer{1}{1}_0$& $\varfield_{1}$ & \\\cline{3-4}
                                           &\headpostab&&&\\\cline{4-4}
                                           \refForward{0}&&\headpostabbuf&&\\
                                           \refJumpRight{1,1,0}&&&\headpostabbuf& \\
                                           \refJoin{1}&&\headpostab&&
                                         \end{tabular}
                                         \caption{}
                            \label{alg:caveat:asif}
                    \end{subfigure}
          \caption{
          (a): Parsing $\alayout \in \RepeatLayouts$, according to the model.
          (b): Parsing a concrete instance of $\alayout' \in \treetoflat(\alayout)$:
               in $\alayout'$ the \reader cannot distinguish $\varfield_1$ and $\varfield_2$, as $\beg{2}$ is never inferred to allow applying \refJoin{2} or \refBackward{2}. This is not sound as $\alayout$ is deserializable, and so should be all the $\alayout' \in \treetoflat(\alayout)$ since  $\streams(\alayout) \supset \streams(\alayout')$.
          (c): Successful parsing of a concrete instance of $\alayout'' \in \treetoflat(\alayout)$.
          }\label{algpic:caveat}
  \end{figure}

\begin{figure}[h]
\begin{algorithm}[H]
 \KwData{$\alayout_0 \in \Layouts$, $V = \{i | \varfield_i \occurs \alayout\}$, $R = \{i | \repetition{}_i \occurs \alayout$ \}}
 \KwResult{A modified knowledge base $\axiomset'$ and the inference graph $\infGraph$}
 Duplicate the content of each repetition in $\alayout_0$ into $\alayout_1$\;
 Build $\axiomset$ according to \autoref{subsub:parserrep} from $\alayout_1$\;
 $\simpleSuccessStruct \leftarrow \forwardChainingInfer(\axiomset)$\;
 \eIf{$\exists . i \in V |  \len{i} \notin \axiomset' \vee \exists . i \in R | \repLen{i} \notin \axiomset' $}{\Return $\simpleFailureStruct$\;}{
 \Return $\simpleFailureStructRet$
 }
 \caption{Deserializability check for enhanced \reader model.}
 \label{alg:complicated}
\end{algorithm}
\end{figure}       

%% file: related.tex
\section{Related Work}\label{sec:related}
  None of the following works uses explicitly, to my knowledge, any Horn clause representation of the parsing task.
 The \erlang language \cite{erlang} has a pattern-matching construct whose patterns can be binary comprehensions \cite{gustafsson2005bit}, similar to list comprehensions in functional programming languages.
 Given a set of bit patterns, the matcher is synthesized by constructing a labeled automaton and expressing the matching as a series of elementary actions: test the size of a field, read bits, test match. 
  The specification of a binary format is subject to the variable binding rules of \erlang; this entails, in practice, that in the case of \autoref{ex:convoluted} one must code the layout manually, make an explicit analysis of the layout, and possibly spreading the definition through several functions or mixed with \erlang statements, reducing the effectiveness of a layout specification as such.
   \packettypes \cite{mccann2000packet} addresses the processing of protocol packets, hence of bit-strings, through a protocol stack;
  \datascript \cite{back2002datascript} is even more concise, describing the language and its features. Both languages have a syntax that is influenced by the C language.
  They offer capabilites such as attaching constraints on fields and their content.
  The constraints can only refer to elements occurring earlier in the stream, thus ruling out instances such as \autoref{ex:convoluted}.
  \padsproj \cite{fisher2011pads} is a framework for analysing and defining bit-level formats; it can generate parsers and serialize data.
  \padsproj can even infer, given a set of binary data supposedly following the same layout, the actual layout and be tolerant with errors, by reporting them and continuing parsing.
  Moreover \cite{fisher2011pads} introduces a general framework to express the semantics of data description languages, focussing on the \emph{types} of fields, where a type represents details such as endianness and encoding of the concrete bitstrings of the field.
  The framework gives the building blocks to create a type system for the data description language.
  Type-correctness then entails parsability of a layout.
  This contrasts with my approach which does not make explicit mention of types of fields, which are not needed for deciding deserializabilty.
Beyond the motivations described in \autoref{sec:intro}, a huge effort in bit-level compilers targets space-efficient exchange formats.
Popular \ascii-based data exchange formats have the advantage of being human-readable (\json) and validable (\xml);
 both do have a wealth of libraries for manipulation with standard interfaces;
 the disadvantage is that ASCII wastes bandwidth -- e.g. encoding a single boolean value in several bytes, instead of a single bit.
 Programming languages have libraries that allow serialization of their data, like in \haskell \cite{cereal,binary} or in C \cite{tpl}, but the definition of the data format is done within the programming language. Data specification languages \cite{avro, protobuf, bson, messagepack, capnproto} allow the definition, processing and evolution of protocol messages and output parser/serializers for several target programming language.
Such products hide the composition of the underlying stream to the user;
 unlike what presented here, the definition language does not allow to decide e.g. where to put a pointer item (see \autoref{sub:layout}), because the packing algorithms that optimize aspects such as alignment and evolvability rely on a predetermined physical layout.

%% file: marskalk.tex
\section{Conclusion and Future Work}\label{sec:marskalk}
  I presented a method to determine whether there is a parser that can parse a stream of bits given a description of the bit layout.
  I introduced a language for describing layouts and I described the behaviour of a \reader as reasoning
   within an untyped first-order logic formal system having axioms in the form of Horn clauses.
   The typical use case of this method is the implementation of a bit-stream data-definition language, or of a serialization library. The benefit is that it enables to use existing Horn inference engines.
   At \cite{mypage} there is a \python \cite{python} implementation of the method using the \clips \cite{clips} rule-based language to perform forward chaining. 
     It defines a language to describe layouts and translate them to a \clips program encoding the axioms, input to the \clips interpreter; the \python script interprets back the output. Using \prolog gives no particular advantages over using other programming languages, since \prolog interpreters do backward chaining reasoning, thus one can either
  implement forward chaining or delegate it to any existing library or external tool.
  The previous sections not discuss any preprocessing of layouts. I report some I observed during the development of this work, which are not closely related with this paper's contribution:
  \begin{inparaenum}[\itshape a)]
  \item 
         Save bits by reducing the value contained in the pointer fields by substituting all $\pointer{o}{s} \occurs \alayout$
         with $\pointer{q}{t}$ such that each pointer range is shrinked enough to begin and end with a variable length field.
          This can be done in linear time with a check on the span of every pointer and updating the labels or spans of the pointers left.
           Once a pointer $p$ is shrinked, one might then redesign manually the layout by reducing the length of $p$. For instance, consider $\alayout = \pointer{0}{5}\varfield_1\field_{\ }\varfield_3\field$. Applying the above optimization results in $\alayout' = \pointer{1}{3}\varfield_1\field_{\ }\varfield_3\field$.
   \item 
         Allow only forward pointers. Backward pointers are unusual in practice, because they can imply buffering that can be avoided.
         One could consider only those layouts such that $\forall \anitem . \big(\anitem  \occurs \alayout\big) 
            \wedge \big(\anitem = \pointer{\intlist{b,a}}{r}_{\intlist{b,x}}\big) \Rightarrow \big( x \leq a \big)$, $a,x \in \naturals$.
         This only constraint does not anyway guarantee that all such layouts are deserializable.
         For instance
         $\alayout' =  \field_{\ } \pointer{2}{4}_{\ } \field_{\ } \varfield_3 \pointer{5}{1} \varfield_5 $
         is not deserializable.
   \item 
           If a pointer's purpose is exclusively to determine the lengths of variable fields, then remove pointers that span over no variable length fields or repetitions.
           This can be done in linear time with a check on the span of every pointer and updating the labels or spans of the pointers left.
  \end{inparaenum}
  More complicated analyses and extensions, which are part of future work, are the following:
  \begin{inparaenum}[\itshape i)]
    \item Permute the fields so that minimal buffering is needed. Consider
            $\alayout =  \field_{\ } \pointer{2}{3}_{\ } \field_{\ } \varfield_3 \varfield_4 \pointer{4}{1}$.
           The layout
          $\alayout' =  \field_{\ } \field_{\ } \pointer{3}{1}_{\ } \varfield_3 \pointer{5}{1} \varfield_5 $ is a permutation of $\alayout$; but in $\alayout$ one must buffer both $\varfield_3$ and $\varfield_4$ before being able to distinguish them.
    \item Have a side-effect free constraint language (like \datascript or \padsproj in \autoref{sec:related} do) to express constraints between values and lengths of fields; the constraints contribute in building the axiom set. 
          Consider layout $\alayout = \field_0\field_1\varfield_2$.
          If $\varfield_2$ is a sequence of samples, $\field_0$ tells the number of samples in $\varfield_2$ and $\field_1$ tells the number of bits each sample has, then this corresponds to the axiom $\val{0} \wedge \val{1} \Rightarrow \len{2}$.This constraint feature enables for instance to use variable fields as pointers.
            \item The inference graph can be used to generate a parser for streams $\streams(\alayout)$.
                   The axioms applied during the reasoning can be translated into actions, similarly to \cite{gustafsson2005bit}:
                   \ref{ax:jumpRight} corresponds to buffering new data from the stream,
                   and \ref{ax:jumpLeft} or again \ref{ax:jumpRight} to addressing within the buffer in case of already buffered data.
                   \ref{ax:join}, \ref{ax:forward} and \ref{ax:backward} correspond to consuming data and associating it to a field.
       Note that it is an optimization problem: since the inference graph is a DAG,
       there are several topological orderings each of which would map to a parser with specific performances in e.g. memory consumption. Describing details of this optimization and related research is future work.    
    \end{inparaenum}